\documentclass{llncs}
\usepackage{subfigure}
\usepackage{pgfarrows,tikz}
\usetikzlibrary{shapes,arrows,oodgraph}
\usepackage[utf8]{inputenc}
\usepackage[T1]{fontenc}
\usepackage{amsmath,amsfonts,amssymb,amscd,mathbbol,graphicx,hyperref,calc,ifthen,enumitem}

\newcommand{\N}{\mathbb{N}}
\newcommand{\R}{\mathbb{R}}
\usepackage[color=red!60,textsize=small]{todonotes}

\title{On the Dynamics of\\ Bounded-Degree Automata Networks}
\author{Julio Aracena\inst4\and Florian Bridoux\inst3\and Maximilien Gadouleau\inst5\and Pierre Guillon\inst1 \and Kévin Perrot\inst2 \and Adrien Richard\inst3 \and Guillaume Theyssier\inst1}
  \institute{Aix-Marseille Université, CNRS, I2M UMR7373, Marseille, France
    \and Aix-Marseille Université, Univ. Toulon, CNRS, LIS UMR7020, Marseille, France
    \and Univ. C\^{o}te d'Azur, CNRS, I3S UMR 7271, Sophia Antipolis, France
    \and CI2MA and Departamento de Ingeniería Matemática, Universidad de Concepción, Chile
    \and Department of Computer Science, Durham University, Durham, UK
}

\newcommand\funs{\mathcal{F}}
\newcommand\dyna{\mathcal{D}}

\newcommand{\ipart}[1]{\left\lfloor #1\right\rfloor}
\newcommand{\card}[1]{\left|#1\right|}
\newcommand{\rank}{\text{rk}}
\newcommand{\inNeighbors}{N^{-}}

\newcommand{\GF}{\mathrm{GF}}
\newcommand{\B}{\{0,1\}}

\def\fp{\mathrm{fp}}

\newcommand\ie{, \textit{i.e.},\ }

\begin{document}
\maketitle

\begin{abstract}
  Automata networks can be seen as bare finite dynamical systems, but their growing theory has shown the importance of the underlying communication graph of such networks.
  This paper tackles the question of what dynamics can be realized up to isomorphism if we suppose that the communication graph has bounded degree.
  We prove several negative results about parameters like the number of fixed points or the rank. 
   We also show that we can realize with degree 2 a dynamics made of a single fixed point and a cycle gathering all other configurations.
  However, we leave open the embarrassingly simple question of whether a dynamics consisting of a single cycle can be realized with bounded degree, although we prove that it is impossible when the network become acyclic by suppressing one node, and that realizing precisely a Gray code map is impossible with bounded degree.
  Finally we give bounds on the complexity of the problem of recognizing such dynamics.
\end{abstract}

\section{Introduction}

One possible definition for automata network is simply a self-map $F:Q^n\to Q^n$.
This definition forgets about the computational aspect of the model, which consists, through a dual point of view, in a set of $n$ automata linked by some arcs, each holding a state in $Q$ that they can update depending on that of their incoming neighbors.

\begin{figure}
    \centering
  \subfigure[a cycle of length $2^n$.]{\begin{tikzpicture}[every node/.style={inner sep=.5pt}]
      \begin{oodgraph}
        \addcycle[xshift=8cm,radius=1.5cm,edges style=-stealth]{16};
      \end{oodgraph}
    \end{tikzpicture}\label{f:2n}}
  \qquad
    \subfigure[a cycle of length $2^n-1$ and an isolated vertex.]{\begin{tikzpicture}[every node/.style={inner sep=.5pt}]
      \begin{oodgraph}
        \addcycle[xshift=4cm,radius=1.4cm,edges style=-stealth]{15};
        \addcycle[xshift=4cm,edges style=-stealth,loop style={out=70, in=110,looseness=16}]{1};
      \end{oodgraph}\end{tikzpicture}\label{f:2n-1}}
\qquad
\subfigure[a cycle of length $2^n-C$ and a size-$C$ forest 
plugged to it.]{\label{f:btrees}
  \begin{tikzpicture}[every node/.style={inner sep=.5pt}]
      \begin{oodgraph}
        \addcycle[xshift=0cm,radius=.8cm,edges style=-stealth,nodes prefix=mycycle]{9};
        \addbeard[attach node=mycycle->6,radius=.5]{1};
        \addbeard[attach node=mycycle->4,radius=.5]{1};
        \addbeard[attach node=mycycle->1,radius=.5]{2};
        \addbeard[attach node=mycycle->1->1,radius=.5]{3};
      \end{oodgraph}\end{tikzpicture}}
    \caption{
      Three examples of dynamics families.
    }
    \label{f:cycles}
\end{figure}
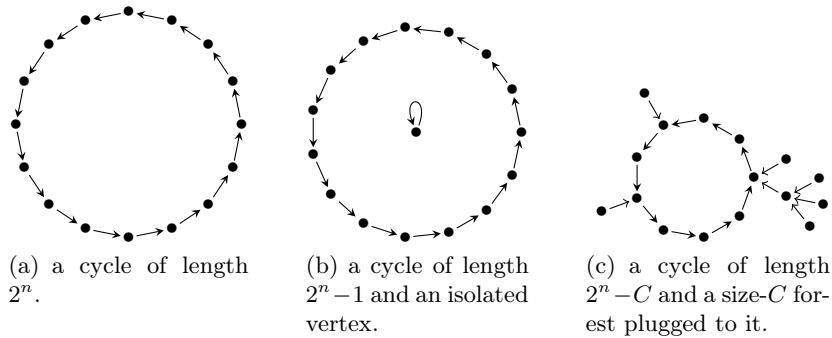

As a model of computation generalizing finite cellular automata, this communication graph is quite relevant, and it is natural to constrain it, in particular the possible degrees: a small degree indeed represents simple local computations, whereas a complete communication graph can yield any dynamics $F: Q^n\to Q^n$ (see \cite{bppmr23}).

The minimal communication graph, often called interaction graph, plays an important role in automata network theory (see \cite{Gadouleau_2019} for a survey).
It was already established that some dynamics require high degree, and even a dense communication graph \cite{bppmr23}.   

In this paper, we address the question of how restrictions on the communication graph, and in particular bounding its degrees, can impose restrictions on the possible dynamics.
For instance, in Figure~\ref{f:cycles}, one can see three (families of) graphs representing possible dynamics.
Which are the ones that can be realized by communication graphs with small degree?

In Section \ref{sec:nonlocdyn}, we establish bounds on different parameters of the dynamics depending on the degree of communication graphs.
This in particular allows to show that the family of dynamics from Figure~\ref{f:btrees} cannot be realized with a bounded-degree communication graph.

In Section~\ref{s:fsr}, we establish positive results through various constructions. In particular, we show that the dynamics of the type from Figure~\ref{f:2n-1} can be realized with communication graphs of degree $2$ (Theorem~\ref{theo:near-hamiltonian}), and prove that we can realize maps of rank ${q^n-2}$ with degree ${2n/3}$ (Theorem \ref{theo:highrank}).

Then, in Section~\ref{sec:booleancase}, we focus on the Boolean case and prove several negative results. First, we show that having Hamiltonian dynamics (a single cycle containing all configurations, see Figure~\ref{f:2n}) requires at least degree $3$ (Theorem~\ref{theorem:no_affine_hamiltonian}). We also study the particular case of communication graphs which become acyclic by removing a single node: we show that, in this case, producing Hamiltonian dynamics requires degree $n$ (Theorem~\ref{thm:even}). Besides, we establish (Theorem~\ref{thm:gray_code}) that a Boolean network requires degree at least $\log(n)$ to produce a Gray code map, which are particular examples of Hamiltonian dynamics (where any configuration is at Hamming distance 1 from its image).

Finally, in Section~\ref{sec:complexity}, we give upper and lower bounds for the computational complexity of recognizing dynamics that can be realized with a bounded-degree communication graph.

However, we leave open the question about the minimum degree necessary to realize Hamiltonian dynamics. Prior to this work, J. Aracena and A. Zapata formulated the conjecture that such dynamics requires unbounded degree. This also appears in \cite{zapata22} with various intermediate results. 

\section{Definitions and notations}

Consider a finite {alphabet} $Q$ with $q=\card Q$ symbols.
Without loss of generality, $Q=\{0,\dots,q-1\}$.
Consider also a set ${V=\{1,\ldots,n\}}$ of $n$ \emph{nodes}.
A \emph{configuration} $x=(x_i)_{i\in V}\in Q^V$ is a function $V \to Q$. 
For every $U\subseteq V$, we denote $x_U:U\to Q$ the restriction of $x$ to $U$\ie $(x_U)_i=x_i$ for every $i \in U$.
Given a \emph{pattern} $u \in Q^U$, we define the \emph{cylinder} 
  $[u] = \{x\in Q^V:x_{U}=u\}$.

An \emph{automata network} (AN) is a map $F:{Q^V\to Q^V}$.
It can be represented as a \emph{dynamics graph}, like those from Figure~\ref{f:cycles}, by linking each configuration $x$ to its image $F(x)$.
This graph is denoted by ${\dyna(F)}$.
Two ANs are called isomorphic if their dynamics graphs are isomorphic.
A configuration $x$ such that ${F(x)=x}$ is called a \textit{fixed point}, and the number of fixed points of $F$ is denoted ${\fp(F)}$.
The set of ANs with alphabet of size $q$ and with $n$ nodes is denoted ${\funs(n,q)}$.

For $F\in\funs(n,q)$ and $k\in\N$, let us define $Y_k=\{y\in Q^V\mid\card{F^{-1}(y)}=k\}$.
We also note $Y_{\ge\ell}=\bigsqcup_{k\ge\ell}Y_k$.
The {\em rank} of $F$ is its number of images: $\rank(F)=Y_{\ge1}$.
Remark that $Q^V=Y_{\ge0}$, and that $\sum_{k\in\N}k\card{Y_k}=\card{F^{-1}(Q^V)}=q^n$, so that $\card{Y_0}=q^n-\rank(F)=\sum_{k\ge1}(k-1)\card{Y_k}=\sum_{k\ge2}(k-1)\card{Y_k}$.
This quantity is bounded between $\card{Y_{\ge2}}$ and $\card{Y_{\ge2}}(\max_{y\in Q^V}\card{F^{-1}(y)}-1)$.
In particular, $\rank(F)<q^n\iff\card{Y_0}>0\iff\card{Y_{\ge2}}>0$.
Configurations from $Y_0$ are sometimes called \emph{orphans}, whereas \emph{collisions} are pairs of configurations with the same image.

\paragraph{Communication graph.}
A \emph{communication graph} for $F$ is a directed graph over vertex set $V$ with edges in $E\subset V^2$ such that for every ${j\in V}$,
and every $x,x'\in Q^V$ which agree over the in-neighborhood $\inNeighbors(j)={\{i\in V\mid (i,j)\in E\}}$ of $j$, we have $F(x)_{j}=F(x')_{j}$.
In other words, the value $F(x)_j$ is updated thanks to a local function $f_j:Q^V\to Q$ which depends only on the values $x_{\inNeighbors(j)}$.
For $U\subseteq V$, we may also denote $f_U(x)=F(x)_U$.
The \emph{interaction graph} of $F$, denoted $G(F)$, is the minimal communication graph of $F$.
Its \emph{degree} is the maximum in-degree of a vertex in $G(F)$.
By extension, the \emph{degree} of $F$ is the degree of its interaction graph.
We denote by ${\funs(n,q,d)}$ the set of ANs from $\funs(n,q)$ with degree at most $d$.

A first remark is that if $u\in Q^{\inNeighbors(j)}$, then $\card{[u]}=2^{n-\card{\inNeighbors(j)}}=2^{n-d}$ if the in-degree of $j$ in the communication graph is $d$.

Adapting this argument to pre-images, we get one key tool for the following section, emphasizing that bounded-degree ANs are very specific dynamical systems.
\begin{lemma}\label{lem:localrigidity}
  Consider ${F\in\funs(n,q,d)}$ and ${U\subseteq V}$ with ${\card U\le\lfloor n/d\rfloor}$.
  Then for any pattern ${u\in Q^U}$, $\card{F^{-1}([u])}$ is a multiple of ${q^{n-\card Ud}}$.
\end{lemma}
\begin{proof}
  Since the degree of $G(F)$ is upper-bounded by $d$, $f_U$ only depends of $Y = \bigcup_{i \in U}\inNeighbors(i)$, so that $\card Y\le\card Ud$.
  In other words, for every $x\in Q^{U}$ such that $f_U(x)=u$, we obtain $f_U([x_Y])=u$.
  Hence, $\card{F^{-1}([u])}=\card{\{v\in Q^{Y}\mid f_U([v])=u\}}q^{n-\card Y}$.
  Since $\card Y\le\card Ud$, this is a multiple of $q^{n-\card Ud}$.
  \qed
\end{proof}

\section{Non-local dynamics}
\label{sec:nonlocdyn}

Here we prove that some kinds of dynamics are intrinsically non-local in the sense that they cannot be realized by bounded-degree networks, even up to isomorphism.

\begin{remark}
  The number of nonisomorphic bijective ANs is  $p( q^n )$ (where $p$ is the partition function), which is asymptotically given by the Hardy-Ramanujan formula (see \cite{Andrews_1984}): 
  \[
    p( q^n ) \sim \frac{ 1 }{ 4 q^n \sqrt{3} } \exp( \pi \sqrt{ 2 q^n /3 } ).
  \]
  which grows doubly exponentially in $n$.
  However, there are only $(q^{q^d})^n$ AN with degree $\le d$, which is simply exponential in $n$.
  So few bijective dynamics have a realization with bounded degree.
\end{remark}

The identity AN on $Q^V$ ($F(x) = x$ for all $x$) has $q^n$ fixed points and degree $1$. Our first result shows that if $G(F)$ has bounded degree and $F$ is not the identity, then the number of fixed points of $F$ cannot be close to $q^n$. 

\begin{proposition}\label{pro:notidentity}
Let ${F\in\funs(n,q,d)}$ with $\fp(F)<q^n$. Then ${\fp(F)\leq q^n - q^{n-d}}$.
\end{proposition}
\begin{proof}
Since $F$ is not the identity map, there exist $i\in V$ and $x \in Q^V$ such that $f_i(x) \neq x_i$. 
  There are two cases. If $i \notin \inNeighbors(i)$, then every pattern $u \in Q^{V \setminus \{i\}}$ admits two extensions $y,y'\in[u]$, with $y_i\ne y'_i$, but $f_i(y)=f_i(y')$, so that at most one of them is a fixed point.
  Hence, ${\fp(F)\leq q^n-q^{n-1} \leq  q^n - q^{n-d}}$.
  On the other hand, if $i \in \inNeighbors(i)$, then let $u = x_{\inNeighbors(i)}$;
  for every configuration $y \in [u]$, $f_i(y) = f_i(x) \neq x_i = y_i$ and $y$ is not a fixed point.
  Therefore, ${\fp(F)\leq  q^n - \card{[u]}\le q^n-q^{n-d}}$.
  \qed
\end{proof}

\begin{remark}
The bound from the previous lemma is tight: indeed let $F(x) = x$ if $x_i\neq 0$ for some $1\leq i\leq d$ and $F(x)=(\pi(x_1),x_2,\ldots,x_n)$ otherwise, where $\pi$ is a derangement of $Q$. Then $F$ is an AN of degree ${d}$ with ${q^n - q^{n-d}}$ fixed points. Alternatively, consider the graph $G$ on $\{1, \dots, n\}$ with arcs $\{ (i,i) : i \in \{1, \dots, n \} \} \cup \{ (i,1) : i \in \{2, \dots, d\}  \}$. Then $G$ has degree $d$ and following \cite[Theorem 3]{GRF16}, there is an AN with interaction graph $G$ and exactly $q^n - q^{n-d}$ fixed points (namely, $F$ given above).
\end{remark}

Proposition~\ref{pro:notidentity} can be generalised to the powers of $F$. First, note that if $F\in\mathcal{F}(q,n,d)$, then $F^k\in \mathcal{F}(q,n,d^k)$ for every $k\geq 1$ (because from $G(F)$ of degree $\leq d$ we obtain a communication graph for $F^k$ by putting an edge for each path of length $k$).
By combining this remark and Proposition~\ref{pro:notidentity}, we obtain that, if $\fp(F^k)<q^n$, then $\fp(F^k)\leq q^n-q^{n-d^k}$. 

As an application, we can easily find bijections without fixed points that force large communication degrees. Suppose for instance that the dynamics of $F\in\funs(2,n)$ consists of $2^{n-1}-2$ limit cycles of length $2$ and one limit cycle of length $4$. Then $F^2$ has exactly $2^n-4$ fixed points. Denoting by $d$ the degree of $G(F)$, we obtain that $2^n-4=\fp(F^2)\leq 2^n-2^{n-d^2}$ and thus $d\geq \sqrt{n-2}$.

Our second result shows that if $G(F)$ has bounded degree and $F$ is not a bijection, then the rank of $F$ cannot be close to $q^n$.
In \cite{zapata22}, it is shown that Boolean dynamics with rank $2^n-1$ require degree $n$.
Our result extends this to all dynamics with rank $q^n - 1$.

\begin{theorem}\label{thm:almostnonbij}
  Let ${F\in\funs(n,q,d)}$ with $\rank(F)<q^n$.
  Then $\rank(F)\le q^n-\frac{n+\log_q2}{d+\log_q2}$.
  In particular, if $d<n$, then $\rank(F)\le q^n-2$.
\end{theorem}
In particular, this theorem implies that the family of dynamics depicted in Figure~\ref{f:btrees} is impossible to realize with bounded-degree ANs. However, Theorem~\ref{thm:almostnonbij} fails among bijective ANs of fixed degree, such as the dynamics depicted in Figure~\ref{f:btrees}, as we will see in Section~\ref{s:fsr}.

To prove the theorem, we need a simple witnessing lemma.
\begin{lemma}\label{l:witness}
  If $Y\subset Q^V$ with $\card Y\le n$, and $x\in Q^V\setminus Y$, then there exists $U\subset V$ with $\card U\le\card Y$ and $[x_U]\cap Y=\emptyset$.
\end{lemma}
\begin{proof}
  Let us prove the statement by induction on $Y$.
  If $Y=\emptyset$, the trivial cylinder with $U=\emptyset$ is suitable.
  Now, let $Y$ and $x$ be such that there exists $U\subset V$ with $\card U\le\card Y$ and $[x_U]\cap Y=\emptyset$.
  Let us prove the statement for $Y\cup\{y\}$, where $y\in Q^V\setminus(Y\cup\{x\})$.
  Since $x\ne y$, there exists $i\in V$ such that $x_i\ne y_i$.
  Note that $\card{U\cup\{i\}}\le\card U+1\le\card Y+1$. 
  Besides, $[x_{U\cup\{i\}}]\subset[x_U]$; by induction hypothesis, it does not intersect $Y$.
  Moreover, $[x_{U\cup\{i\}}]\subset[x_{\{i\}}]\not\ni y$.
  It results that $[x_{U\cup\{i\}}]\cap(Y\sqcup\{y\})=\emptyset$.
\qed\end{proof}

Here is now the key lemma for lower-bounding the orphans.
\begin{lemma}\label{lem:balancebij}
  Let $F\in\funs(n,q)$.
  If $\card{Y_0}\ge1$ and $\card{Y_{\ge2}}\le\ipart{n/d}$, then $\card{Y_0}\ge q^{n-\card{Y_{\ge2}}d}$.
\end{lemma}
\begin{proof}
  Let $x\in Y_0$.
  Since $x\notin Y_{\ge2}$, Lemma~\ref{l:witness} gives some $U\subset Q^V$ such that $\card U\le\card{Y_{\ge2}}\le\ipart{n/d}$ and $[x_u]\cap Y_{\ge2}=\emptyset$.
  One can write $\card{F^{-1}([x_U])}$ as $\card{F^{-1}([x_U]\cap Y_0)}+\card{F^{-1}([x_U]\setminus Y_0)}$.
  The first term is $0$, by definition of $Y_0$, and the second is $\card{[x_U]\setminus Y_0}$, by nonintersection with $Y_{\ge2}$.
  Since $x\in[x_U]\cap Y_0$, $\card{F^{-1}([x_U])}=\card{[x_U]\setminus Y_0}<q^{n-\card U}$.
  On the other hand, Lemma~\ref{lem:localrigidity} allows to write $\card{F^{-1}([x_U])}$  as $\alpha q^{n-\card Ud}$, for some $\alpha\in\N$.
  Since $\alpha q^{n-\card Ud}<q^{n-\card U}$, we get that $\alpha\le q^{\card U(d-1)}-1$.
  Putting things together, $\card{[x_U]\setminus Y_0}=\card{F^{-1}([x_U])}\le(q^{\card U(d-1)}-1)q^{n-\card Ud}=\card{[x_U]}-q^{n-\card Ud}$.
    We get that $\card{Y_0}\ge\card{[x_U]}-\card{[x_U]\setminus Y_0}\ge q^{n-\card Ud}\ge q^{n-\card{Y_{\ge2}}d}$.
\qed\end{proof}

\begin{proof}[of Theorem~\ref{thm:almostnonbij}]
If $\card{Y_{\ge2}}>\ipart{n/d}$, then $\card{Y_0}\ge\card{Y_{\ge2}}>\ipart{n/d}$ and we are done.
  Otherwise, Lemma~\ref{lem:balancebij} gives that $\card{Y_0}\ge q^{n-\card{Y_{\ge2}}d}\ge q^{n-\card{Y_0}d}$.
  Hence, $\log_q\card{Y_0}\ge n-\card{Y_0}d$.
  The map $x\mapsto(x-1)\log_q2$ upper-bounds $\log_q$ in the whole set $\R_+^*\setminus]1,2[\supset\N$. 
  So we get that $(\card{Y_0}-1)\log_q2\ge\log_q\card{Y_0}\ge n-\card{Y_0}d$.
  Hence, $\card{Y_0}\ge\frac{n+\log_q2}{d+\log_q2}$.
  Whenever $d<n$, this fraction is strictly more than $1$, so that the integer $\card{Y_0}$ has to be at least $2$.
  \qed
\end{proof}
Our bound on $\card{Y_0}$ could be slightly improved by involving better upper bounds on $\log_q$, such as $x\mapsto(\log_q(m+1)-\log_qm)(x-m)+\log_qm$, or $x\mapsto(x-1)^{1/m}\log_q2$, with $m\in\N^*$.
Nevertheless, the resulting expressions become less easy to manipulate.

Here is another application of Lemma~\ref{lem:localrigidity}. 

\begin{proposition}
  Let ${F\in\funs(n,q,d)}$ such that $F$ is not constant.
  Then the number of preimages of any configuration is upper-bounded by ${q^n - q^{n-d}}$.
\end{proposition}
\begin{proof}
Let $y \in Q^V$.
Let us prove that $|F^{-1}(y)|\leq q^n - q^{n-d}$.
Since $F$ is not constant, there exists $z\in F(Q^V)$ such that $z_i \neq y_i$ for some $i \in V$. Since $F^{-1}([z_i]) \neq \emptyset$, by Lemma~\ref{lem:localrigidity}, $| F^{-1}([z_i]) | \geq q^{n-d}$. Furthermore, since $F^{-1}([z_i]) \cap F^{-1}(y) = \emptyset$, we conclude $|F^{-1}(y)| \leq q^n - q^{n-d}$.\qed
\end{proof}

\begin{remark}
  It is tight because we can have $F(x) = (0,\ldots,0)$ if $x_i\neq 0$ for some $1\leq i\leq d$ and $F(x)=(1,0,\ldots,0)$ otherwise.
\end{remark}

\section{Realization results}\label{s:fsr}

\subsection{Feedback shift registers}
  In this section, we are interested in realizing examples of AN with \emph{almost degree $1$}\ie~whose all but one nodes have degree at most $1$.

  We use the following important tool.
  Let $g:Q^n\to Q$, and $F_g:Q^n\to Q^n$ be the corresponding \emph{feedback shift register} ({FSR}), that is, $F_g(x)=F_g(x_1,\ldots,x_{n})=(x_2,\dots,x_{n},g(x))$.
  $G(F_g)$ is thus obtained from the path $1\leftarrow 2\leftarrow\cdots\leftarrow n$ by adding an arc from $i$ to $n$ whenever $g$ depends on input $i$: it has almost degree $1$.
  
The \emph{de Bruijn graph} of order $n$ over the alphabet $Q$ has vertex set ${V=Q^n}$ and arc set ${E=\{(au,ub):a,b\in Q,u\in Q^{n-1}\}}$.

\begin{proposition}[\cite{Lempel71}]\label{p:debruijn}
  For any $n$ and any ${1\leq k\leq q^n}$, the de Bruijn graph of order $n$ admits a cycle of length $k$.
\end{proposition}

\begin{proposition}
  For any $n$ and any ${1 \leq k\leq q^n}$, there exists ${F:Q^n\to Q^n}$ with almost degree $1$ and whose maximum limit cycle has length $k$.
\end{proposition}
\begin{proof}
  Consider some cycle $C$ of length $k$ in the de Bruijn graph of order $n$ over $Q$, given by Proposition~\ref{p:debruijn}, and the feedback shift register $F_g$, where
  \[{g(x) =
      \begin{cases}
        b&\text{ if $x=au$ and $au\to ub\in C$};\\
        0&\text{ otherwise.}
      \end{cases}
    }\]
  $F_g$ has almost degree $1$, and has the cycle $C$ in its dynamics.
  To conclude the proof, it is sufficient to observe that the dynamics on the complement of $C$ consists in adding $0$ at node $n$ and shifting node $i+1$ to node $i$ for ${i<n}$.
  Therefore, the only possible cycle created by this part of the dynamics is possibly the fixed point ${(0,\ldots, 0)}$.\qed
\end{proof}

\subsection{Construction of near-Hamiltonian dynamics with in-degree $2$}

We say that $F: Q^V \to Q^V$ is near-Hamiltonian if it has one fixed point and a cycle of length $q^n - 1$. In this section, we let $q$ be a prime power and $Q = \GF(q)$ be the finite field of order $q$. We can then construct a near-Hamiltonian AN $F: \GF(q)^n \to \GF(q)^n$ with an interaction graph of maximum in-degree $2$.

\begin{theorem}\label{theo:near-hamiltonian}
    For any prime power $q$ and any $n \ge 2$, there exists a near-Hamiltonian AN in $\funs( n, q, 2 )$.
\end{theorem}

\begin{proof}
Let $\GF(q^n)$ be generated by the primitive polynomial $P(\xi)  = \sum_{i=0}^{n-1} p_i \xi^i$ and let $\alpha$ be a root of $P(\xi)$, i.e. a primitive element of $\GF(q^n)$. We then identify $\GF(q^n)$ and $\GF(q)^n$ as follows:
\[
    x = (x_0, x_1, \dots, x_{n-1}) \in \GF(q)^n \sim \beta = x_0 + x_1 \alpha + \dots + x_{n-1} \alpha^{n-1} \in \GF( q^n ).
\]
Then 
\[
    F:x\mapsto\alpha x
\]
is near-Hamiltonian: $0^n$ is a fixed point, and since $\GF(q^n)^*$ is a cyclic group generated by $\alpha$, we have the cycle $1 \mapsto \alpha \mapsto \dots \mapsto \alpha^{q^n - 2} \mapsto \alpha^{q^n - 1} = 1$.

For any $\beta \in \GF(q^m)$, we have
\begin{align*}
    \alpha \beta &= \alpha \sum_{j=0}^{m-1} x_j \alpha^j \\
    &= \sum_{j=0}^{m-2} x_j \alpha^{j+1} + x_{m-1} \alpha^m \\
    &= \sum_{i=0}^{m-1} x_{i-1} \alpha^i +  x_{m-1} \sum_{i=0}^{m-1} (-p_i) \alpha^i \\
    &= \sum_{i=0}^{m-1} ( x_{i-1} - p_i x_{m-1}  ) \alpha^i.
\end{align*}
The local functions are then given by
\[
    f_i(x) = x_{i-1} - p_i x_{m-1},
\]
(with $x_{-1} = 0$), hence $F$ has degree $2$.
  \qed
\end{proof}

\begin{example}
Let $q = 2$, $n = 3$, $P( \xi ) = \xi^3 + \xi + 1$. Then $\alpha^3 = \alpha + 1$, and 
\[
    \GF(2^3) = \{ 0, 1, \alpha, \alpha^2, \alpha^3 = \alpha + 1, \alpha^4 = \alpha^2 + \alpha, \alpha^5 = \alpha^2 + \alpha + 1, \alpha^6 = \alpha^2 + 1 \}
\]
(and $\alpha^7 =1$). We identify $\GF(2)^3$ and $\GF(2^3)$ as follows:
\begin{align*}
    000 & \sim 0 \\
    100 & \sim 1 \\
    010 & \sim \alpha \\
    001 & \sim \alpha^2 \\
    110 & \sim \alpha + 1 = \alpha^3 \\
    011 & \sim \alpha^2 + \alpha = \alpha^4 \\
    111 & \sim \alpha^2 + \alpha + 1 = \alpha^5 \\
    101 & \sim \alpha^2 + 1 = \alpha^6.
\end{align*}
Then $F$ is given as follows: 
\[
    0 \mapsto 0, \quad 1 \mapsto \alpha \mapsto \alpha^2 \mapsto \alpha^3 \mapsto \alpha^4 \mapsto \alpha^5 \mapsto \alpha^6 \mapsto 1, 
\]
or from a Boolean network point of view:
\[
    000 \mapsto 000, \quad 100 \mapsto 010 \mapsto 001 \mapsto 110 \mapsto 011 \mapsto 111 \mapsto 101 \mapsto 100. 
\]
In terms of local functions, we have
\begin{align*}
    f_0(x) &= x_2 \\
    f_1(x) &= x_0 + x_2 \\
    f_2(x) &= x_1.
\end{align*}
\end{example}

\subsection{Rank $q^n - 2$ with $d = 2n/3$}

We have shown in Theorem \ref{thm:almostnonbij} that rank $q^n - 1$ required degree $n$. For rank $q^n - 2$, however, the bound only yields $d \ge n/2 - 1$; we now prove that we can achieve $d =  \lceil \frac{2}{3}n  \rceil$.

\begin{theorem}\label{theo:highrank}
For all $n \ge 3$ and all odd $q \ge 3$, there exists an AN in $\funs(n, q, d = \lceil \frac{2}{3}n  \rceil )$ with rank $q^n - 2$.
\end{theorem}


\begin{proof}
First, we consider the case $n = 3$ and $d = 2$. Let $q \ge 3$ be odd and $Q = \mathbb{Z}_q$. The local functions of $F$ are given as follows.

\begin{align*}
    f_1( x ) &= \begin{cases}
        x_1 & \text{if } x_1 \ge 2 \\
        (x_1 + x_3) \mod 2 & \text{if } x_1 \in \{ 0, 1 \},
    \end{cases} \\
    f_2( x ) &= \begin{cases}
        x_2 &\text{if } x_1 \ge 2 \\
        (x_1 + x_2) \mod q &\text{if } x_1 \in \{ 0, 1 \}, x_2 \ne 0 \\
        1 &\text{if } x_1 x_2 = 00 \\
        0 &\text{if } x_1 x_2 = 10,
    \end{cases} \\
    f_3( x ) &=\begin{cases}
        x_3 &\text{if } x_2 \ne 0 \\
        ( x_3 + 1 ) \mod q &\text{if } x_2 = 0.
    \end{cases}
\end{align*}

For instance, for $q=3$ we obtain:
\begin{center}
  \begin{minipage}{.3\linewidth}
    \begin{tabular}{c|c}
      $x$ & $F(x)$ \\
      \hline
      $000$ &  $011$ \\
      $001$ &  $112$ \\
      $002$ &  $010$ \\
      $010$ &  $010$ \\
      $011$ &  $111$ \\
      $012$ &  $012$ \\
      $020$ &  $020$ \\
      $021$ &  $121$ \\
      $022$ &  $022$ \\
    \end{tabular}
  \end{minipage}
  \begin{minipage}{.3\linewidth}
  \begin{tabular}{c|c}
    $x$ & $F(x)$ \\
    \hline
    $100$ &  $101$ \\
    $101$ &  $002$ \\
    $102$ &  $100$ \\
    $110$ &  $120$ \\
    $111$ &  $021$ \\
    $112$ &  $122$ \\
    $120$ &  $100$ \\
    $121$ &  $001$ \\
    $122$ &  $102$ \\
\end{tabular}
\end{minipage}
  \begin{minipage}{.3\linewidth}
\begin{tabular}{c|c}
    $x$ & $F(x)$ \\
    \hline
    $200$ &  $201$ \\
    $201$ &  $202$ \\
    $202$ &  $200$ \\
    $210$ &  $210$ \\
    $211$ &  $211$ \\
    $212$ &  $212$ \\
    $220$ &  $220$ \\
    $221$ &  $221$ \\
    $222$ &  $222$ 
\end{tabular}
\end{minipage}
\end{center}
We now search for collisions. One can easily check the following two collisions:
\begin{align}
    \label{equation:first_collision}
    F(0, 0, q-1) &= F(0, 1, 0), \\
    \label{equation:second_collision}
    F(1, 0, q-1) &= F(1, q-1, 0).
\end{align}
We now prove that those are the only collisions. Suppose that $a = a_1a_2a_3$ and $b = b_1b_2b_3$ are two distinct configurations, say $q^2 a_1 + q a_2 + a_3 < q^2 b_1 + q b_2 + b_3$, such that $F(a) = F(b)$. We proceed by a case analysis.

\begin{enumerate}
    \item $b_1 \ge 2$. \\
    Then $f_1( a ) = f_1( b ) \ge 2$, hence $a_1 = f_1( a ) = f_1( b ) = b_1 \ge 2$. Moreover, $f_2( 
a ) = a_2 = f_2( b ) = b_2$. Thus $a_2 = b_2$ and $a_3 \ne b_3$, which yields $f_3( a ) \ne f_3( b )$, which is the desired contradiction.

    \item $a_1 = 0$, $b_1 = 1$. \\
    Since $f_2( a ) = f_2( b )$, we obtain $a_2 \in \{ 2, \dots, q-1 \}$ and $b_2 = a_2 - 1 \in \{1, \dots, q-2 \}$. Since $f_3( a ) = f_3( b )$ and $a_2, b_2 \ne 0$, we obtain $a_3 = b_3$. But then $f_1( a ) = a_3 \mod 2 \ne (b_3 + 1) \mod 2 = f_3( b )$, which is the desired contradiction.

    \item $a_1 = b_1 = 0$. \\
    First, we have $a_2 \ne b_2$, since otherwise $a_2 = b_2$ and $a_3 \ne b_3$ thus $f_3( a ) \ne f_3( b )$. Now, since $f_2( a ) = f_2( b )$, we obtain $a_2 = 0$ and $b_2 = 1$. Then $(a_3 + 1) \mod q = f_3( a ) = f_3( b ) = b_3$ and $a_3 \mod 2 = f_1( a ) = f_1( b ) = b_3 \mod 2$; those two constraints are both satisfied only if $a_3 = q-1$ and $b_3 = 0$. Therefore $a$ and $b$ are the collision in \eqref{equation:first_collision}. 

    \item $a_1 = b_1 = 1$. \\
    Again, we have $a_2 \ne b_2$, hence $a_2 = 0$ and $b_2 = q-1$. By the same reasoning as above, we obtain $a_3 = q-1$ and $b_3 = 0$. Therefore $a$ and $b$ are the collision in \eqref{equation:second_collision}.
\end{enumerate}

Having proved the case $n = 3$, we now move on to the case where $n = 3\ell$ for some $\ell \ge 1$. Let $q$ be odd and let $k = q^\ell$ be odd as well. Consider $F \in \funs(k, 3, 2)$ as described above. By identifying $\mathbb{Z}_k$ with $( \mathbb{Z}_q )^\ell$, we obtain a network $\tilde{F} \in \funs(q, n = 3\ell, d = 2\ell)$ of rank $k^3 - 2 = q^n - 2$.

We now deal with the other case, say $n = 3\ell + r$ for some $r \in \{1, 2\}$; write $V = L \cup R$ with $L = \{1, \dots, 3\ell\}$ and $R = \{3\ell + 1, \dots, n\}$. Let $\tilde{F} \in\funs(q, 3\ell, 2\ell)$ of rank $q^{3\ell} - 2$ as above. Then let $\hat{F} \in\funs(q, n = 3\ell + r, d = 2\ell + r )$ choose between the identity function on $L$ or $\tilde{F}$, depending on the control bits in $R$:
\begin{align*}
    \hat{F}( x )_L &= \begin{cases}
        x_L & \text{if } x_R \ne 0 \\
        \tilde{F}( x_L ) & \text{if } x_R = 0
    \end{cases} \\
    \hat{F}( x )_R &= x_R.
\end{align*}
Then $\hat{F}(x) = \hat{F}(y)$ for some $x \ne y$ if and only if $x_R = y_R = 0$ and $x_L$ and $y_L$ collide: $\tilde{F}( x_L ) = \tilde{F}( y_L )$. Thus there are only two collisions.
  \qed
\end{proof}

\section{Further negative results for the Boolean case}
\label{sec:booleancase}

In this section, we consider the case of $q=2$ and $d=2$.
This is an interesting case, due to the important role taken by affine functions. 
A function $f : Q^n \to Q$ is \emph{affine} if $Q = \GF(q)$ and $f( x ) = a \cdot x + b$ for some vector $a \in \GF(q)^n$ and some scalar $b \in \GF(q)$. 
An AN is affine if all its local functions $f_i$ are affine, i.e. if $Q = \GF(q)$ is a finite field and $F(x) = Ax + v$ for some matrix $A \in \GF( q )^{n \times n}$ and some vector $v \in \GF( q )^n$.

For $q = 2$, $\GF(2)$ is a finite field and in this case $+$ denotes the addition modulo $2$ (and we use the same notation for sum of configuration modulo $2$ component-wise).
We begin by a folklore lemma on balanced Boolean functions of at most two variables, whose proof we give for the sake of completeness.

\begin{lemma} \label{lemma:balanced_q=2_d=2}
All the balanced functions $f: \{0,1\}^2 \to \{0,1\}$ (that is, functions such that $|f^{-1}(0)|=|f^{-1}(1)|$) are affine.
\end{lemma}

\begin{proof}
The map $f \mapsto f^{-1}( 0 )$ is a bijection between the set of balanced functions $f : \{0,1\}^2 \to \{0,1\}$ to the pairs of elements in $\{0,1\}^2$. 
As such, there are $\binom{4}{2} = 6$ balanced functions $f : \{0,1\}^2 \to \{0,1\}$. 
Conversely, the $6$ affine functions of the form $f(x_1, x_2) = a_1 x_1 + a_2 x_2 + a_0$ for $(a_1, a_2) \ne (0,0)$  are clearly balanced.
\end{proof}

In particular, any permutation in $\funs(n,q=2,d=2)$ has balanced local functions, and hence must be affine.

\begin{corollary} \label{corollary:affine}
Any permutation in $\funs(n,q=2,d=2)$ must be affine.
\end{corollary}

This algebraic restriction leads to strong dynamical restrictions, as seen below.

\subsection{Non-existence of Boolean Hamiltonian with degree $2$}

We call an AN $F : Q^V \to Q^V$ Hamiltonian if its dynamics consists of a single cycle of length $q^n$. We prove that $\funs(n,2,2)$ does not contain any Hamiltonian AN (for $n \ge 3$).

\begin{theorem} \label{theorem:no_affine_hamiltonian}
If $F$ is an affine AN over $\GF(q)^n$ with $n \ge 3$, then it is not Hamiltonian.
\end{theorem}

\begin{proof}
Computer search settles the case where $n=3$ and $q=2$. We now assume $(n,q) \ne (3,2)$, which is equivalent to $n \le q^{n-2}$.

Let $F$ be affine, i.e. $F(x) = Ax + v$ for some matrix $A \in \GF( q )^{n \times n}$ and some vector $v \in \GF( q )^n$. For the sake of contradiction, suppose that $F$ is Hamiltonian. Denoting $k = q^n$, we have
\[
    F^k(x) = A^k x + ( A^{k-1} + A^{k-2} + \dots + A + I )v = x,
\]
hence $A^k = I$.

Let $B = A - I$. Since $A$ and $-I$ commute, we have $B^k = A^k + (-I)^k$ \cite[Theorem 1.46]{LN97} and hence $B^k = A^k + (-1)^k I = A^k - I = 0$. Thus $B$ is nilpotent and by simple linear algebra, $B^n = 0$. Since $n \le q^{n-2}$, we have $B^{ q^{n-2} } = 0$, and hence $A^{ q^{n-2} } = I$.

Thus $F^{ q^{n-2} }( x ) = x + u$ for some vector $u$ and
\[
    F^{ q^{n-1} }( x ) = x + qu = x,
\]
which contradicts the fact that $F$ is Hamiltonian.
\qed
\end{proof}

\begin{corollary}
Let $q=2$ and $n \ge 3$. If $F$ is Hamiltonian, then $F$ has degree at least $3$.
\end{corollary}

\begin{proof}
Suppose $F$ is Hamiltonian of degree $2$. By Corollary \ref{corollary:affine} $F$ is affine, which contradicts Theorem \ref{theorem:no_affine_hamiltonian}.
\qed
\end{proof}

\subsection{Upper bound on the rank}

We can significantly refine the bound in Theorem \ref{thm:almostnonbij} for the case $q = 2$, $d = 2$.

\begin{proposition}
Suppose $F \in \funs(n,2,2)$ with $\rank( F ) < 2^n$. Then $\rank( F ) \le 2^n - 2^{ n-2 }$.
\end{proposition}

\begin{proof}    
  Suppose that $F$ is a non-bijective AN in $\funs(n,2,2)$.
  First, if all its local functions are balanced, then $F$ is affine, hence $\rank(F) \le 2^{n-1} < 2^n - 2^{n-2}$.
  Second, if the local function $f_i$ is not balanced, then there exists $b \in \{0,1\}$ such that $|f_i^{-1}(b)| \ge 3$.
  Let $A = f_i^{-1}(b)$ and $B = \{ x : x_i = b \}$.
  Denoting $\bar{A} = \{0,1\}^V \setminus A$, we obtain
\[
    | f( \{0,1\}^V ) | \le | f( A ) | + | f( \bar{A} ) | \le | B | + | \bar{A} | \le 2^{n-1} + 2^{n-2} = 2^n - 2^{n-2}.
\]
\qed
\end{proof}

\begin{remark}
This bound is also tight. Indeed, let $F \in \funs(n,2,2)$ be defined by $f_1( x ) = x_1 \land x_2$ and $f_i( x ) = x_i$ otherwise. Then $F( \{0,1\}^V ) = \{ x \in \{0,1\}^V : x_1 x_2 \ne 10 \}$ so that $\rank( F ) = 2^n - 2^{n-2}$.
\end{remark}

\subsection{Hamiltonian dynamics on centralized interaction graphs}\label{sec:centralized}

To be more concise in the following, we denote by $e_i$ the configuration equal to $1$ at node $i$ and $0$ elsewhere.

Aracena and Zapata conjecture in \cite{zapata22} that there is a constant $\alpha>0$ such that, for every $n\in\N$, if $F\in\funs(n,2)$ is Hamiltonian, the degree of $F$ is at least $\alpha n$.
In this section, we prove (a strong variant of) the conjecture under the assumption that $G(F)$ is {\em centralized}, that is, $G(F)$ has a node whose deletion makes the graph acyclic.
In the following, we abusively say that $F$ is centralized when $G(F)$ is. Note that FSRs are centralized.
So we will prove that there is no centralized Hamiltonian function in $\funs(n,2,d)$ when $d<n$.
We actually prove something stronger. 
\begin{theorem}\label{thm:even}
Let $F\in \funs(n,2,d)$ be a centralized bijection. If $n\geq 3$ and $d<n$, then $F$ has an even number of limit cycles. 
\end{theorem}

The main tool is a swap operation on $F$, taken from \cite{F82}, defined (in our setting) as follows. Given distinct $x,y\in\B^n$, let $(x\leftrightarrow y)$ the permutation of $\B^n$ that swaps $x$ and $y$: $(x\leftrightarrow y)(x)=y$, $(x\leftrightarrow y)(y)=x$ and $(x\leftrightarrow y)(z)=z$ for all $z\neq x,y$. Given  $F\in\funs(n,2)$, we say that $F'=F\circ (x\leftrightarrow y)$ is a {\em swap} of $F$. Let $p(F)\in\B$ be the parity of the number of limit cycles in $F$, and suppose that $F$ is a bijection. Then $F'$ is a bijection and the swap operation changes the parity of the number of limit cycles: $p(F')\neq p(F)$. Indeed, let $C_x$ and $C_y$ be the limit cycles of $F$ containing $x$ and $y$, respectively, and let $\ell$ and $\ell'$ be the numbers of limit cycles in $F$ and $F'$, respectively. Clearly every limit cycle of $F$ distinct from $C_x,C_y$ is also a limit cycle of $F'$. Then, we have two cases. First, if $C_x=C_y$, then the swap operation splits this limit cycle into two limit cycles so that $\ell'=\ell+1$; see Figure \ref{fig:swap}(a) for an illustration. Second, if $C_x\neq C_y$, then the swap operation joins the two limit cycles into one limit cycle so that $\ell'=\ell-1$; see Figure \ref{fig:swap}(b) for an illustration. Thus in any case $p(F)\neq p(F')$.  

\begin{figure}
\[\tag{a}
\begin{array}{ccc}
\begin{tikzpicture}[every node/.style={inner sep=.5pt,outer sep=1}]
\node (1) at (0:1){$y$};
\node (2) at (45:1){\scriptsize$\bullet$};
\node (3) at (90:1){\scriptsize$\bullet$};
\node (4) at (135:1){\scriptsize$\bullet$};
\node (5) at (180:1){$x$};
\node (6) at (225:1){\scriptsize$\bullet$};
\node (7) at (270:1){\scriptsize$\bullet$};
\node (8) at (315:1){\scriptsize$\bullet$};
\path[->,thick]
(1) edge (2)
(2) edge (3)
(3) edge (4)
(4) edge (5)
(5) edge (6)
(6) edge (7)
(7) edge (8)
(8) edge (1)
;
\end{tikzpicture}
&\qquad&
\begin{tikzpicture}[every node/.style={inner sep=.5pt,outer sep=1}]
\node (1) at (0:1){$y$};
\node (2) at (45:1){\scriptsize$\bullet$};
\node (3) at (90:1){\scriptsize$\bullet$};
\node (4) at (135:1){\scriptsize$\bullet$};
\node (5) at (180:1){$x$};
\node (6) at (225:1){\scriptsize$\bullet$};
\node (7) at (270:1){\scriptsize$\bullet$};
\node (8) at (315:1){\scriptsize$\bullet$};
\path[->,thick]
(1) edge (6)
(2) edge (3)
(3) edge (4)
(4) edge (5)
(5) edge (2)
(6) edge (7)
(7) edge (8)
(8) edge (1)
;
\end{tikzpicture}
\\[1mm]
F&&F'=F\circ (x\leftrightarrow y)
\end{array}
\]

\[\tag{b}
\begin{array}{ccc}
\begin{tikzpicture}[every node/.style={inner sep=.5pt,outer sep=1}]
\node (1) at (0:1){$y$};
\node (2) at (45:1){\scriptsize$\bullet$};
\node (3) at (90:1){\scriptsize$\bullet$};
\node (4) at (135:1){\scriptsize$\bullet$};
\node (5) at (180:1){$x$};
\node (6) at (225:1){\scriptsize$\bullet$};
\node (7) at (270:1){\scriptsize$\bullet$};
\node (8) at (315:1){\scriptsize$\bullet$};
\path[->,thick]
(1) edge (6)
(2) edge (3)
(3) edge (4)
(4) edge (5)
(5) edge (2)
(6) edge (7)
(7) edge (8)
(8) edge (1)
;
\end{tikzpicture}
&\qquad&
\begin{tikzpicture}[every node/.style={inner sep=.5pt,outer sep=1}]
\node (1) at (0:1){$y$};
\node (2) at (45:1){\scriptsize$\bullet$};
\node (3) at (90:1){\scriptsize$\bullet$};
\node (4) at (135:1){\scriptsize$\bullet$};
\node (5) at (180:1){$x$};
\node (6) at (225:1){\scriptsize$\bullet$};
\node (7) at (270:1){\scriptsize$\bullet$};
\node (8) at (315:1){\scriptsize$\bullet$};
\path[->,thick]
(1) edge (2)
(2) edge (3)
(3) edge (4)
(4) edge (5)
(5) edge (6)
(6) edge (7)
(7) edge (8)
(8) edge (1)
;
\end{tikzpicture}
\\[1mm]
F&&F'=F\circ (x\leftrightarrow y)
\end{array}
\]
{\caption{\label{fig:swap} Illustration of the swap operation.}}
\end{figure}
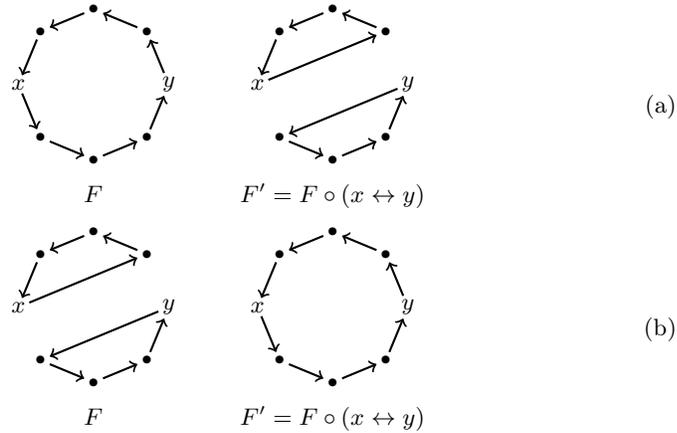

More generally, for $k\geq 1$, we say that $F'$ is a {\em $k$-swap} of $F$ if there exists configurations $x^1,y^1,\dots,x^k,y^k$, with $x^i\neq y^i$ for all $1\leq i\leq n$, such that 
\[
F'=F\circ (x^1\leftrightarrow y^1)\circ\cdots\circ(x^k\leftrightarrow y^k). 
\]
By convention, the $0$-swap of $F$ is $F$ itself. The $k$-swap operation preserves the bijectivity, and since each individual swap changes the parity of the number of limit cycles, we obtain the following lemma. 

\begin{lemma}\label{lem:swap}
Let $F\in\funs(n,2)$ be a bijection, and let $F'$ be a $k$-swap of $F$. Then $p(F)=p(F')$ if and only if $k$ is even. 
\end{lemma}

In \cite{F82}, Fredricksen gives a survey of Hamiltonian FSRs and, given a bijective FSR $F\in\funs(n,2)$, the swap operation is used to connect $p(F)$ and the {\em weight} of $F$ defined (in our setting) as the number $w(F)$ of configurations $x\in\B^n$ such that $x_n<f_1(x)$. Let $\sigma\in F(n,2)$ be the {\em circular shift}, defined by 
\[
\sigma(x)=(x_n,x_1,\dots,x_{n-1}).
\]
Fredricksen proves that $\sigma$ is a $w(F)$-swap of $F$. He also says, without proof, that $p(\sigma)=0$. From these two properties and Lemma~\ref{lem:swap}, we obtain that $p(F)$ is the parity of $w(F)$. An unmentioned and easy to prove consequence is that if node $1$ has in-degree at most $n-1$ in $G(F)$, then  $w(F)$ is even (this will be generalized in Lemma~\ref{lem:even_weight}) and thus $F$ has an even number of limit cycles: this proves Theorem~\ref{thm:even} (and thus Aracena-Zapata's conjecture) for FSRs. 

In addition to this simple observation, our contribution is an extension of the mentioned results to centralized networks, giving Theorem~\ref{thm:even}. We start by giving a simple proof that $\sigma$ has an even number of limit cycles (when $n\geq 3$), which already used the swap technique.

\begin{lemma}\label{lem:even_sigma}
For all $n\geq 3$, we have $p(\sigma)=0$. 
\end{lemma}

\begin{proof}
Let $\bar\sigma\in\funs(n,2)$ defined by $\bar\sigma(x)=\sigma(x)+e_1$. Since $\sigma(x+y)=\sigma(x)+\sigma(y)$ we have, $\bar\sigma^2(x)=\sigma(\bar\sigma(x))+e_1=\sigma(\sigma(x)+e_1)+e_1=\sigma^2(x)+\sigma(e_1)+e_1$. More generally, for all $k\geq 1$, 
\[
\bar\sigma^k(x)=\sigma^k(x)+\sigma^{k-1}(e_1)+\cdots+ \sigma^0(e_1). 
\]
In particular, since $\sigma^n$ is the identity, we have 
\[
\bar\sigma^n(x)=x+\sigma^{n-1}(e_1)+\cdots+ \sigma^0(e_1)=x+e_n+e_{n-1}+\cdots+e_1=x+1^n,
\]
so $\bar\sigma^n$ is the negation.
Suppose that $\bar\sigma$ has exactly $r$ limit cycles, with length $c_1,\dots,c_r$.
Since $\bar\sigma^n$ is the negation, $\bar\sigma^{2n}$ is the identity.
Thus $\bar\sigma$ is a bijection and each $c_i$ divides $2n$ but not $n$. Let $\alpha\geq 0$ be the integer such that $n/2^{\alpha}$ is odd; since $n\geq 3$ we have $n>\alpha+1$.
The fact that $\ell_i$ divides $2n$ but not $n$ means that $c_i=2^{\alpha+1} q_i$ for some odd integer $q_i$. Since $\bar\sigma$ is a bijection, we have 
\[
2^{\alpha+1}\sum_{i=1}^r q_i=\sum_{i=1}^r\ell_i=2^n
\]
and thus $q_1+\cdots+q_r=2^{n-\alpha-1}\geq 2$. Since every $q_i$ is odd we deduce that $r$ is even, that is, $p(\bar\sigma)=0$. Let $X$ be the set of configurations $x\in\B^n$ with $x_n=0$, and let  $x^1,\dots,x^k$ be an enumeration of $X$, so $k=2^{n-1}$. Let $F$ be the $k$-swap of $\bar\sigma$ defined by 
\[
F=\bar\sigma\circ (x^1\leftrightarrow x^1+e_n)\circ\cdots\circ(x^k\leftrightarrow x^k+e_n). 
\]
For all $x\in X$ we have $\sigma(x+e_n)=\sigma(x)+e_1$; hence $
F(x)=\bar\sigma(x+e_n)=\sigma(x+e_n)+e_1=\sigma(x)$ and $
F(x+e_n)=\bar\sigma(x)=\sigma(x)+e_1=\sigma(x+e_n)$. Thus $F=\sigma$ is a $k$-swap of $\bar\sigma$ and since $k$ is even, by Lemma~\ref{lem:swap}, $p(\sigma)=p(\bar\sigma)=0$.
\qed
\end{proof}

We now extend the notion of weight to the centralized case. We need the following property. 

\begin{lemma}\label{lem:Hamiltonian}
If $F\in \funs(n,2)$ is a centralized bijection, then $G(F)$ is Hamiltonian. 
\end{lemma}

\begin{proof}
Let $F\in \funs(n,2)$ be a bijection. Gadouleau proves in \cite{G18rank} that $G(F)$ contains a spanning subgraph which is a disjoint union of cycles. In $G(F)$, this spanning subgraph necessarily consists of a single cycle, and thus $G(F)$ is Hamiltonian. 
\qed
\end{proof}

Let $F\in \funs(n,2)$ be a centralized bijection, and let $C$ be a Hamiltonian cycle in $G(F)$.
Let $i\in V$ and let $j$ bet its in-neighbor in $C$.
We denote by $w_i(F,C)$ the number of configurations $x\in\B^n$ such that $x_j<f_i(x)$, and we set
\[
w(F,C)=\sum_{i=1}^n w_i(F,C).
\]
Note that if $F$ is a FSR, there is a unique Hamiltonian cycle $C$ (whose vertices are $1,2\dots,n$ in order) and since $w_i(F,C)=0$ for all $i\neq 1$, we have $w(F)=w_1(F,C)=w(F,C)$ and we recover the previous definition. 

Let $\sigma^C\in\funs(n,2)$ defined by: for all $i\in V$ and $x\in\B^n$, $\sigma^C_i(x)=x_j$ where $j$ is the in-neighbor of $i$ in $C$. Obviously, $\sigma^C$ is isomorphic to $\sigma$ and has thus an even number of limit cycles, and $\sigma^C=\sigma$ when the vertices of $C$ are $1,2,\dots,n$ in order. That $\sigma$ is a $w(F)$-swap of a bijective FSR $F$ is then generalized as follows.

\begin{lemma}\label{lem:swap_to_sigma}
Let $F\in\funs(n,2)$ be a centralized bijection and $C$ a Hamiltonian cycle of $G(F)$. Then $\sigma^C$ is a $w(F,C)$-swap~of~$F$. 
\end{lemma}

\begin{proof}
Suppose without loss that the vertices of $C$ are $1,2,\dots,n$ in order, so that $\sigma^C=\sigma$. For all $i\in V$ and $x\in\B^n$, we have 
\begin{equation}\label{eq:wi}
w_i(F,C)=0~\Rightarrow~ f_i(x)=x_{i-1}
\end{equation}
where $x_0$ means $x_n$. Indeed, let $X$ be the set of $x\in\B^n$ with $x_{i-1}=0$. Since $w_i(F,C)=0$, if $x_{i-1}=0$, then $f_i(x)=0$. Hence $X\subseteq f^{-1}_i(0)$. Since $F$ is a bijection we have $|f^{-1}_i(0)|=|f^{-1}_i(1)|=2^{n-1}$ and since $|X^0|=2^{n-1}$ we deduce that $f^{-1}_i(0)=X$. Consequently, if $x_{i-1}=1$, then $f_i(x)=1$. This proves \eqref{eq:wi}. 

We now prove, by induction on $w(F)$, that $\sigma$ is a $w(F,C)$-swap of $F$. If $w(F,C)=0$, then $F=\sigma$ by \eqref{eq:wi}.
This proves the base case.
For the induction, suppose that $w(F)>0$. Since each node $i$ in $G(F)$ with $w_i(F,C)=0$ is, by \eqref{eq:wi}, of in-degree one, and since $G(F)$ is centralized, there is a node $i$ with $w_i(F,C)>0$ whose deletion leaves $G(F)$ acyclic.
Suppose, without loss of generality, that this node is node $1$. Then $1$ is the unique out-neighbor of $n$ since otherwise there is a cycle which does not contain node~$1$. We deduce that, for all $x\in \B^n$,  
\begin{equation}\label{eq:n1}
F(x+e_n)=F(x)+e_1.
\end{equation}
Indeed, since $1$ is the unique out-neighbor of $n$, $F(x+e_n)$ and $F(x)$ differ at most in component $1$, and since $F$ is a bijection this forces $F(x+e_n)=F(x)+e_1$. Let $y\in\B^n$ such that $y_n<f_1(y)$, which exists since $w_1(F,C)>0$, and let  
\[
F'=F\circ (y\leftrightarrow y+e_n). 
\]
Then $F'(y)=F(y+e_n)=F(y)+e_1$, and thus $F'_1(y)=0$. Furthermore, $F'(y+e_n)=F(y)=F(y+e_n)+e_1$ and for all $x\in\B^n$ with $x\neq y,y+e_n$ we have $F'(x)=F(x)$. Hence $G(F')$ has an arc from $n$ to $1$, and since $f'_i=f_i$ for all $i\neq 1$ we deduce that $C$ is contained in $G(F')$, and that $w(F',C)=w(F,C)-1$. By induction, $\sigma$ is a $w(F',C)$-swap of $F'$ and since $F'$ is a $1$-swap $F$ we deduce that $\sigma$ is a $w(F,C)$-swap of $F$.  
\qed
\end{proof}

Putting things together, we obtain the following.
\begin{lemma}\label{lem:same_parity}
Let $F\in\funs(n,2)$ be a centralized bijection and let $C$ be a Hamiltonian cycle of $G(F)$. Then $p(F)$ is the parity of $w(F,C)$. 
\end{lemma}

\begin{proof}
By Lemma~\ref{lem:swap_to_sigma}, $\sigma^C$ is a $w(F,C)$-swap of $F$. Since $\sigma^C$ is isomorphic to $\sigma$, by Lemma~\ref{lem:even_sigma} we have $p(\sigma^C)=p(\sigma)=0$. By Lemma~\ref{lem:swap} $p(F)=0$ if and only $w(F,C)$ is even. Thus $p(F)$ is the parity of $w(F,C)$. 
\qed
\end{proof}

To conclude, we need the following easy lemma.

\begin{lemma}\label{lem:even_weight}
Let $F\in\funs(n,2,d)$ be a centralized bijection and let $C$ be a Hamiltonian cycle of $G(F)$. If $d<n$, then $w(F,C)$ is even. 
\end{lemma}

\begin{proof}
  Suppose that $d<n$.
  Let $i\in V$ and $j$ be its in-neighbor in $C$.
  Let $X$ be the set of $x\in \B^n$ with $x_j<f_i(x)$.
  Thus $|X|=w_i(F,C)$. Since $d<n$, there exists $k\in V$ such that $G(F)$ has no arc from $k$ to $i$. Let $x\in X$. Since there is an arc from $j$ to $i$ we have $k\neq j$ thus $(x+e_k)_j=0$, and since there is no arc from $k$ to $i$ we have $f_i(x+e_k)=f_i(x)=1$, thus $x+e_k\in X$. We deduce that $x\in X$ if and only if $x+e_k\in X$, which proves that $|X|=w_i(F,C)$ is even. Thus each $w_i(F,C)$ is even, and so is $w(F,C)$. 
\qed
\end{proof}

The proof of Theorem~\ref{thm:even} is now straightforward. 

\begin{proof}[of Theorem \ref{thm:even}]
Let $F\in\funs(n,2,d)$ be a centralized bijection. By Lemma~\ref{lem:Hamiltonian}, $G(F)$ has a Hamiltonian cycle $C$. If $d<n$, then $w(F,C)$ is even by Lemma~\ref{lem:even_weight} and thus $p(F)=0$ by Lemma~\ref{lem:same_parity}. 
\qed
\end{proof}

Theorem~\ref{thm:even} suggests the following strengthening of Aracena-Zapata's conjecture: if $F\in \funs(n,2,d)$ is bijective and $d<n$, then $F$ has an even number of limit cycles. 

\subsection{Gray codes}

A {\em Gray code} is an enumeration of the configurations in $\B^n$ such that two successive configurations differ in one component, and such that the first and last ones also differ in one component. Gray codes are well known structures with many applications \cite{savage1997survey}. In our setting, a Gray code is a Hamiltonian function $F\in\funs(n,2)$ such that, for all $x\in \B^n$, $x$ and $F(x)$ differ in one component. In this section, we prove the following.
All the logarithms are here in base $2$.

\begin{theorem}\label{thm:gray_code}
If $\funs(n,2,d)$ contains a Gray code, then $d\geq \log n$.
\end{theorem}

This provides a proof, for Gray codes, of the following weaker form of Aracena-Zapata's conjecture: for any fixed $d$, if $n$ is large enough, then $\funs(n,2,d)$ has no Hamiltonian function. 

For the proof we need the following definitions.
Let $\delta(x,y)$ be the {\em Hamming distance} between $x$ and $y$, that is, the number of vertices $i\in V$ such that $x_i\neq y_i$.
Given $F\in\funs(n,2)$ we set
\[
\delta(F)=\sum_{x\in\B^n} \delta(x,F(x)).
\]
If $F$ is a Gray code, then $\delta(F)=2^n$.
Given $i\in V$, let us say that $f_i$ is a {\em trivial component} of $F$ if $f_i$ is constant or $f_i(x)=x_i$ for all $x\in\B^n$.
For instance, if $F$ is a bijection with an odd number of limit cycles, then $F$ has no trivial component. 

\begin{lemma}\label{lem:delta}
Let $0<\varepsilon\leq 1$ and $F\in\funs(n,2)$ without trivial component. If $\delta(F)\leq n^{(1-\varepsilon)}2^n$, then $G(F)$ has at least $\varepsilon n\log n$ arcs.
\end{lemma}

\begin{proof}
  Suppose that $\delta(F)\leq n^{(1-\varepsilon)}2^n$.
  Let $X_i$ be the set of $x\in\B^n$ with $F_i(x)\neq x_i$. Note that $X_i$ is non-empty since otherwise $f_i$ is a trivial component. Note also that 
\begin{equation}\label{eq:delta}
\sum_{i=1}^n |X_i|=\delta(F)\leq n^{(1-\varepsilon)}2^n.
\end{equation}
If $i\not\in \inNeighbors(i)$, then for all $x\in\B^n$, we have $f_i(x)=f_i(x+e_i)$.
Thus exactly one among $x,x+e_i$ belongs to $X_i$, and thus $|X_i|=2^{n-1}\geq 2^{n-d_i}$ since $d_i\geq 1$ (because $f_i$ is not constant). Suppose that $i\in \inNeighbors(i)$, and let $x\in X_i$. For any $y$ with $y_{\inNeighbors(i)}=x_{\inNeighbors(i)}$ we have $y_i=x_i\neq f_i(x)=f_i(y)$ so $y\in X_i$, and we deduce that $|X_i|\geq 2^{n-d_i}$. Thus in any case
\[
d_i\geq n-\log |X_i|.
\] 
Hence the number $e$ of arcs in $G(F)$ is 
\[
e=\sum_{i=1}^n d_i\geq n^2-\sum_{i=1}^n \log |X_i|=n^2-\log \Big(\prod_{i=1}^n |X_i|\Big).
\]
Using the  AM-GM inequality and then \eqref{eq:delta} we have 
\[
\prod_{i=1}^n |X_i|\leq \left(\frac{\sum_{i=1}^n|X_i|}{n}\right)^n\leq \left(\frac{n^{(1-\varepsilon)}2^n}{n}\right)^n=2^{n^2-\varepsilon n\log n}.
\]
We deduce that 
\[
e\geq n^2-\log (2^{n^2-\varepsilon n\log n})=\varepsilon n\log n.
\]
\qed
\end{proof}

\begin{proof}[of Theorem \ref{thm:gray_code}]
Let $F\in\funs(n,2,d)$ be a Gray code. Since $F$ has no trivial component, and since $\delta(F)=2^n$, by Lemma~\ref{lem:delta} (applied with $\varepsilon=1$), $G(F)$ has at least $n\log n$ arcs, and thus the average in-degree is $\log n\leq d$.  
\qed
\end{proof}

\section{Complexity of recognizing bounded-degree dynamics}
\label{sec:complexity}

Fix $d$ and $q$, and consider the following decision problem called BDD$_{q,d}$ (bounded-degree dynamics):
given ${F\in\funs(n,q)}$ represented by Boolean circuits (each $f_i$ is represented by a circuit), is there some ${F'\in\funs(n,q,d)}$ such that their dynamics graphs, ${\dyna(F)}$ and ${\dyna(F')}$, are isomorphic?

\begin{theorem}\label{theo:BDD}
  The problem BDD$_{q,d}$ is in PSPACE for every $d,q$, and co-NP-hard for any $q\ge2$ and ${d\geq 1}$.
\end{theorem}
\begin{proof}
  For the upper bound, a naive algorithm solving BDD$_{q,d}$ consists in guessing
  ${F'\in\funs(n,q,d)}$ (which can be represented by $n$ Boolean circuits of constant size thanks to the bounded-degree condition)
  and checking that ${\dyna(F)}$ and ${\dyna(F')}$ are isomorphic.
  Given that planar graph isomorphism is computable with a LOGSPACE Turing machine M~\cite{dlntw10}
  and that ${\dyna(F)}$ and ${\dyna(F')}$ are at most exponentially larger than the input (Boolean circuit for $F$), we can test isomorphism of ${\dyna(F)}$ and ${\dyna(F')}$ in PSPACE by simulating each reading step of the read-only input tape of M by a polynomial-time circuit evaluation (testing $F(x)=y$ is the same as testing the presence of the corresponding arc in $\dyna(F)$, and same for $F'$).
  This gives an algorithm in NP with an oracle in PSPACE\ie an algorithm in the complexity class PSPACE.

  For the co-NP-hardness, we reduce from UNSAT. Given a propositional formula $\phi$
  over $p$ variables $v_1,\dots,v_p$, we construct $F\in\funs(n,q)$ on $n=d+1+p$ automata,
  with $D=\{1,\dots,d+1\}$, $P=\{d+2,\dots,n\}$, noting $V=P\cup D$.
  Let $Q=\{0,\dots,q-1\}$, $w\in Q^D$ a specific pattern (e.g., $0^{d+1}$), and for $u\in Q^P$, consider the valuation $\theta(u)$
  sending each $0$ to false and other symbols to true.
  Set the local functions to be the identity $f_i(x)=x_i$ for every $i\in V\setminus\{1\}$, and:
  \[
    f_1(x)=\begin{cases}
      x_1+1 \bmod q &\text{if } x_D=w \text{ and }\phi(\theta(x_P)),\\
      x_1 &\text{otherwise}.
    \end{cases}
  \]
  If $\phi$ is unsatisfiable, then $F$ is the identity;
  hence it is a positive instance of BDD$_{q,d}$.
  If on the other hand $\phi$ is satisfiable, then $F$ is not the identity, but the set $\fp(F)$ of fixed points includes $Q^V\setminus[w]$, so that
  \[\card{\fp(F)}\ge q^n-q^{n-d-1}>q^n-q^{n-d}.\]
  Proposition~\ref{pro:notidentity} then implies that it is a negative instance of BDD$_{q,d}$.\qed
\end{proof}

If we remove the isomorphism condition from the above problem,
we get another one, called BDIG (bounded-degree interaction graph):
given ${F\in\funs(n,q)}$ represented by Boolean circuits,
is there some ${F'\in\funs(n,q,d)}$ such that ${\dyna(F)=\dyna(F')}$?
or, equivalently, is the degree of the interaction graph of $F$ bounded by $d$?

\begin{theorem}
  The problem BDIG is co-NP-complete.
\end{theorem}
\begin{proof}
  The lower bound is given by the same reduction as in the proof of Theorem~\ref{theo:BDD}.
  For the upper bound, a simple co-NP algorithm consists in guessing an automaton $i\in V$,
  $d+1$ configurations $x^1,\dots,x^{d+1}$, and $d+1$ distinct automata $i_1,\dots,i_{d+1}$,
  then checking for each $j\in\{1,\dots,d+1\}$ that $f_i(x^j)\neq f_i(x^j+e_{i_j})$.
  For each $j$, it checks whether $x^j$ witnesses the effective dependency of $i$ on automaton $i_j$.
  It is possible to guess $d+1$ such witnesses if and only if
  the interaction graph of $F$ has degree at least $d+1$.\qed
\end{proof}

\section{Conclusion}
Automata networks are minimalist objects that nevertheless allow the study of many phenomena characterizing complex systems, including the emergence of global properties that are difficult to predict from local properties. A central idea of this work is that bounding the in-degree of the interaction graph is a natural way to enforce a strong form of locality on the update functions, by limiting the number of inputs each automaton can depend on. Studying how this structural constraint affects the global dynamics falls within the broader investigation of local-to-global phenomena.

This is a particularly strong constraint, as a simple counting argument shows that almost no automata network satisfies a bounded degree condition. Yet, despite its naturalness and relevance, little is known about its dynamical consequences. This paper offers a first step in this direction, and we believe that this line of investigation deserves to be pursued much further. In particular, our results naturally give rise to two open problems that we see as important steps.

The first is whether Hamiltonian dynamics are degree-bounded.

The second concerns the complexity of recognizing bounded degree dynamics. We have shown that the problem is in PSPACE and co-NP-hard. We believe that the lower bound offers more room for improvement than the upper bound. It is even tempting to conjecture that the problem is PSPACE-complete, which would suggest that bounding the degree has a subtle impact on dynamics. A modest step in this direction would be to show that the problem is NP-hard.

\section{Declarations}

\textbf{Ethical statements.} \textit{Not applicable.}\\

\textbf{Competing interest.}  \textit{Not applicable.}\\

\textbf{Authors' contributions.} \textit{Contribution to be considered equal among all authors, alphabetical order used.}\\

\textbf{Funding.} \textit{This work was supported by ECOS-ANID project C19E02 between France and Chile,
  ANID-Chile through BASAL FB210005,
  ANR-24-CE48-7504 ALARICE,
  ANR-18-CE40-0002 FANs, 
HORIZON-MSCA-2022-SE-01 101131549 ACANCOS, and
STIC AmSud CAMA 22-STIC-02 (Campus France MEAE).}\\

\textbf{Special thanks.} \textit{The authors thank  Anahí Gajardo, Diego Maldonado, and Christopher Thraves, who participated in elaborating some ideas for the presented results.
Figure~\ref{f:cycles} is built with Ti\textit{k}Z library \href{http://fgt.i3s.unice.fr/}{OODGraph}.}

\bibliographystyle{plain}
\bibliography{bd_an_dynamics}

\end{document}